\documentclass[journal]{IEEEtran}
\ifCLASSINFOpdf
\else
\fi

\usepackage[utf8]{inputenc}     
\usepackage[english]{babel}

\usepackage{graphicx}        
\usepackage{amsmath,amssymb}    
\usepackage{amsthm}             
\usepackage{mathtools}          
\usepackage{enumitem}           
\usepackage{listings}           
\usepackage{todonotes}          
\usepackage{hyperref}           

\usepackage{amsthm}
\usepackage{amsmath}
\usepackage{amscd}
\usepackage{amssymb}
\usepackage{mathrsfs}
\usepackage{pdfpages}
\usepackage{setspace}
\usepackage{physics}
\usepackage{tcolorbox}

\usepackage{xcolor}
\usepackage{tikz}
 \usepackage{url}
 \usepackage{babel} 
\usepackage{csquotes}
\usepackage[f]{esvect}
\usepackage{blindtext}
\usepackage{hyperref}
\usepackage{bm}
\newcommand{\R}{\mathbb{R}}

\newcommand{\N}{\mathcal{N}}
\newcommand{\E}{\mathbb{E}}

\newcommand{\g}{\mathfrak{g}}

\newcommand{\ga}{\left<}
\newcommand{\dr}{\right>}

\newcommand{\rk}{\mbox{rk}}

\newcommand{\invW}{\rotatebox[origin=c]{180}{W}}
\newtheorem{theorem}{Theorem}[section]

\newtheorem{lemma}[theorem]{Lemma}

\begin{document}
\title{Spectral Initialization for High-Dimensional Phase Retrieval with Biased Spatial Directions}
\author{Pierre~Bousseyroux, Marc Potters
\thanks{Pierre Bousseyroux is a student of Ecole Normale Supérieure de Paris, France and at Chair of Econophysics and Complex Systems,
\'Ecole Polytechnique, Palaiseau, France (e-mail : pierre.bousseyroux@ens.psl.eu)}
\thanks{Marc Potters is with Capital Fund Management, 23 rue de l'Universit\'e, 75007 Paris, France.}}

\maketitle

\begin{abstract}
We explore a spectral initialization method that plays a central role in contemporary research on signal estimation in nonconvex scenarios. In a noiseless phase retrieval framework, we precisely analyze the method's performance in the high-dimensional limit when sensing vectors follow a multivariate Gaussian distribution for two rotationally invariant models of  the covariance matrix $\vb{C}$. In the first model $\vb{C}$ is a projector on a lower dimensional space while in the second it is a Wishart matrix. Our analytical results extend the well-established case when $\vb{C}$ is the identity matrix. Our examination shows that the introduction of biased spatial directions leads to a substantial improvement in the spectral method's effectiveness, particularly when the number of measurements is less than the signal's dimension. This extension also consistently reveals a phase transition phenomenon dependent on the ratio between sample size and signal dimension. Surprisingly, both of these models share the same threshold value.
\end{abstract}


%
\IEEEpeerreviewmaketitle

\section{Introduction}

In various physical measurement scenarios, such as optical systems employing CCD cameras and photosensitive films, it is often only possible to measure the power spectral density, essentially the squared magnitude of a signal's Fourier transform. Unfortunately, these devices cannot capture the phase information of the incoming light, which encodes essential structural details. This challenge of reconstructing a signal solely from its partial information is termed 'phase retrieval.' It has a rich historical background and finds applications in multiple fields.

This problem can be framed mathematically as follows. Let an unknown vector $\vb{x} \in \mathbb{R}^{N}$ (or $\mathbb{C}^{N}$) be 'probed' with $T$ real (or complex) vectors $\vb{a}_k$, in the sense that the measurement apparatus gives $m_k = \left|\left\langle \vb{a}_k, \vb{x} \right\rangle\right|^2$, called intensity measurements, where $k$ ranges from $1$ to $T$.  The phase retrieval problem is

\begin{equation}
\vb{\hat{x}} = \underset{\vb{x}}{\text{argmin}} \left(\sum_{k=1}^T \left(\left|\left\langle \vb{a}_k, \vb{x} \right\rangle\right|^2 - m_k\right)^2\right).
\end{equation}

It belongs to a broad class of generalized linear estimation problems where the $m_k$ are generated with respect to

\begin{equation}\label{P_0}
    P_0(.|\text{ }|\vb{x}_k|)
\end{equation}where $P_0(.|.)$ denotes a conditional distribution modeling a possibly randomized output channel.

While phase retrieval is a non-convex optimization problem with many local minima, it has become very popular in recent years to pursue convex relations of this problem. In this context, a significant breakthrough was achieved with the PhaseLift approach \cite{candes2013phaselift}, which establishes that, under mild conditions and with Gaussian sensing vectors $\vb{a}_k$, PhaseLift can probabilistically recover $\vb{x}$ exactly (up to a global phase factor) when the number of measurements is of the order of $N\log N$.

To establish a mathematically rigorous theory of phase retrieval, it is crucial to investigate whether there exists a unique solution to the noiseless phase-retrieval problem. This has turned out to be quite a challenging question. The real case is solved: \cite{balan2006signal} proved that $T = 2N-1$ is necessary and sufficient to have a unique solution. Using algebraic geometry, \cite{balan2006signal} showed that, in the complex case, the vector $\vb{x}\in \mathbb{C}^N$ is uniquely determined by the $T$ phaseless measurements as soon as $T\geq 4N-2$, but it is a priori NP-hard \cite{fickus2014phase}. Is the bound $4N-2$ necessary? In other words, is the minimum number of measurements necessary for injectivity, denoted as $M(N)$, equal to $4N-2$? In 2013, Heinosaari, Mazzarella, and Wolf (\cite{heinosaari2013quantum}) demonstrated subtle lower bounds for $M(N)$ using embedding theorems from differential geometry. In 2014, the authors of \cite{bandeira2014saving} provided a characterization of injectivity and conjectured that $M(N) = 4N-4$, which they proved for $N = 2, 3$. In 2015, \cite{conca2015algebraic} showed that $M(N)\geq 4N-4$ when $N = 2^k+1$. However, Cynthia Vinzant disproved the conjecture $4N-4$ for $N = 4$ \cite{vinzant2015small}. The value of $M(N)$ remains an open question. Let's keep in mind that the regime where $q := \frac{N}{T}$ is fixed while $N$ and $T$ get large appears to be interesting.

The oldest reconstruction algorithms \cite{gerchberg1972practical} were iterative: they started with a random initial guess of $\vb{x}$ and attempted to refine it using various heuristics. While these algorithms have empirically been seen to succeed in numerous cases, they can also get stuck at stagnation points, the existence of which is due to the non-convex nature of the problem. \cite{waldspurger2018phase} demonstrated that if $T\geq C N$ for a sufficiently large constant $C$, alternating projections succeed with high probability, provided they are carefully initialized. Alternating projections, introduced by \cite{gerchberg1972practical}, is the most ancient algorithm for phase retrieval. Wirtinger Flow is a gradient descent algorithm proposed by \cite{candes2015phase}, which comes with a rigorous theoretical framework. Initialization is a crucial aspect of non-convex optimization to avoid local minima.

As of today, the best-known polynomial time algorithm for the phase-retrieval problem is the approximate message-passing algorithm (AMP), generalized as the Generalized AMP (GAMP) algorithm, introduced in \cite{schniter2014compressive}.

However, most of the algorithms mentioned above require an initialization $\vb{y}$ that is correlated with the true signal $\vb{x}$ in the sense that the overlap $\rho$, defined as $\rho:= \frac{|\left\langle \vb{x}, \vb{y} \right\rangle|^2}{||\vb{x}||^2 ||\vb{y}||^2}$, is strictly positive. Maximizing the overlap $\rho$ is called the weak-recovery problem. For this purpose, spectral methods are widely employed.  Consider the following matrix:

\begin{eqnarray}
\vb{M} &:=& \frac{1}{T} \vb{H} \vb{D_0} \vb{H}^\dagger\\
&=&\frac{1}{T} \sum_{k=1}^T f(m_k)\vb{a}_k \vb{a}_k^\dagger
\end{eqnarray}
where $\vb{a}_k$ are the columns of $\vb{H}$, and $\vb{D}_0 = \text{diag}\left(f(m_k)\right)_{1\leq k\leq T}$ with $f$ being a certain function.

One can assume that the eigenvector of the largest eigenvalue of $\vb{M}$ provides a good estimate of $\vb{x}$. Why is it a good idea? It's a well-established fact that optimization methods are effective in generating solutions. Our objective is to determine the extent to which we should move in the direction of $\vb{a}_k$ given $m_k$. In other words, we need to select the value of $\left\langle \vb{a}_k, \vb{y} \right\rangle$. Due to the symmetry $\vb{y} \to -\vb{y}$ (or $\vb{y} \to \mathrm{e}^{\mathrm{i}\theta}\vb{y}$), we actually need to decide on $|\langle \vb{a}_k, \vb{y} \rangle|^2$ based on the information from $m_k$. It is natural to consider

\begin{equation}\label{quantity}
\frac{1}{T}\sum_{k=1}^T f(m_k) \left|\left\langle \vb{a}_k, \vb{y} \right\rangle\right|^2
\end{equation}
where $f(m_k)$ are some weights that depend on $m_k$, and $f$ has to be chosen. We aim to maximize this quantity, hoping that the maximum will be a good candidate for our phase retrieval problem. \eqref{quantity} can be expressed as $\left\langle \vb{y}, \vb{M}\vb{y} \right\rangle$, and we know that the unit vector which maximizes $\left\langle \vb{y}, \vb{M}\vb{y} \right\rangle$ is the eigenvector corresponding to the largest eigenvalue of $\vb{M}$ since $\vb{M}$ is symmetric (or Hermitian).

There has been an extensive amount of work on phase retrieval with a random matrix $\vb{H}$. The optimal function $f$ and the overlap $\rho$ between $\vb{x}$ and the largest eigenvector of $\vb{M}$ are deeply analyzed in \cite{luo2019optimal}, \cite{lu2020phase} in a general context where $P_0$, introduced in \eqref{P_0}, is different from $\delta(y - |z|^2)$ while $\vb{H}$ is Gaussian. Interestingly, $\rho$ exhibits a phase transition in the sense that when $q$ is beyond a certain value $q_c$ called the weak reconstruction threshold, $\rho$ becomes zero. This weak-recovery transition was previously identified in \cite{mondelli2018fundamental}, where it is stated that $q_c = 1$ for the complex Gaussian matrix scenario and $q_c = 2$ in the case of a real-valued Gaussian matrix. This is reminiscent of how the injectivity thresholds are $q = \frac{1}{4}$ and $q = \frac{1}{2}$ in the complex and the real case, respectively. The phenomenon may be explained by the complex problem having twice as many variables but the same number of equations as the real problem, implying a need for twice as much data in the complex case. 

From now on we will only consider the real case but the methods below can be applied to the complex case with similar results. We take the vectors $\vb{a}_1, ..., \vb{a}_T$ to be independent and identically distributed random multivariate Gaussian variables with a normalized covariance $\vb{C}$ such that $\tau(\vb{C}) = 1$. The ratio $q$ defined by 
$q := \frac{N}{T}$ will be maintained finite while $N$ and $T$ go to infinity. The goal will be to study the overlap $\rho$, defined as $\rho:= \frac{\left\langle \vb{x}, \vb{y} \right\rangle^2}{||\vb{x}||^2 ||\vb{y}||^2}$, where $\vb{y}$ is the eigenvector correspond to the largest eigenvalue of 

\begin{equation}
    \vb{M}:= \frac{1}{T} \sum_{k=1}^T f\left(\left\langle \vb{a}_k, \vb{x} \right\rangle^2\right) \vb{a}_k \vb{a}_k^T.
\end{equation}

The erratum \cite{erratum} of \cite{potters2020first} gives a parametric equation of the curve $(q, \rho(q))$ when $\vb{C} = \vb{1}$ using a framework of free random matrices:

\begin{equation}\label{formula}
    \left\{
    \begin{array}{ccc}
        q(Z) &=& \frac{I_1(Z)}{Z} - I_2(Z)\\
    \rho(Z)&=& \frac{Z^2 I_2'(Z) + I_1(Z)}{Z^2 I_2'(Z) + I_1(Z)-ZI_1'(Z)}
    \end{array}\right.
\end{equation}
for $Z \geq Z^* = 1$, corresponding to $q_c = 2$, where

\begin{eqnarray}
    I_1(Z) &=& \E\left[\frac{a^2 f^2(a^2)}{Z - f(a^2)} + a^2 f(a^2)\right]\\
    I_2(Z) &=&\E\left[\frac{f(a^2)}{Z - f(a^2)}\right]
\end{eqnarray}with $a$ a standard normal variable. We can also read in this erratum that the optimal function wich maximizes the slope of $\rho$ at $0$ is $y\to 1 - \frac{1}{y}$. This framework where $\vb{C} = \vb{1}$ and $f:y\mapsto 1 - \frac{1}{y}$ will be referred to as "the classical case" and $\rho$ will be denoted as $\rho_{cc}$. In this case, note that $I_1$ and $I_2$ can be precisely calculated by introducing the complementary error function erfc.

In our work, $\vb{C}$ is no longer equal to the identity matrix. In other words, the $\vb{a}_k$ remain independent of each other but are no longer rotationally invariant, implying a preference for certain directions.

The rest of the paper is organized as follows. First, we study a case in which the vectors $\vb{a}_k$ are chosen to be orthogonal (rather than being random multivariate Gaussian). We will then study projector matrices $\vb{P}_{\alpha}$, which have two eigenvalues, 0 and $\alpha$. It turns out that we can readily enhance $\rho$ by using such a covariance matrix $\vb{P}_{\alpha}$ that constrains the vectors $\vb{a}_k$ to a subspace of dimension smaller than $N$. Next, we will calculate $\rho$ when the covariance is a Wishart matrix $\vb{C}$, extending the formula \eqref{formula} continuously. This result will be illustrated by numerical simulations. Finally, we will prove that the two models, where $\vb{C} = \vb{C}_\alpha$ or a Wishart matrix, strangely exhibit the same threshold value $q_c$. Section \ref{section:proofs} will be dedicated to the proofs of these results.

\textbf{Conventions.} We use bold capital letters for matrices and bold lowercase letters for vectors. Some useful tools of RMT are summarized in Appendix \ref{app:appendixRMT}.

\section{Main Results}

\subsection{The orthogonal case}\label{SectionA}

One might have the intuition that to optimally recover the vector $\vb{x}$ we should explore as many directions as possible, therefore orthogonal probing vectors $\vb{a}_k$ (possible when when $T\leq N$) might be better than random ones. However, the following theorem reveals that the spectral method is ineffective in this case.

\begin{theorem}\label{orthogonalcase}
Let $\vb{a}_1, ..., \vb{a}_T$ represent the first $T$ columns of an $N\times N$ Haar distributed random orthogonal matrix. Let $\vb{M}$ be
\begin{eqnarray}
    \frac{1}{T} \sum_{k=1}^T f\left(m_k\right) \vb{a}_k \vb{a}_k^T
\end{eqnarray}with $f$ an increasing function and $m_k:=\ga \vb{a}_k, \vb{x} \dr^2$.

Then, in the large $N$ and $T$ limit while $q :=\frac{N}{T}$ remains finite,
\begin{equation}
    \rho(q) := \E\left[\frac{\ga \vb{x}, \vb{y}\dr^2}{||\vb{x}||^2 ||\vb{y}||^2}\right] = 0
\end{equation}
where $\vb{y}$ the eigenvector corresponding to the largest eigenvalue of $\vb{M}$.
\end{theorem}

The subtlety lies in the fact that we do not know the sign of $\ga \vb{x}, \vb{a}_k\dr$. Contrastingly, if the $\vb{a}_k$ are nearly orthogonal (allowing for substantial space exploration), one might anticipate that the information from other $y_j = \ga \vb{x}, \vb{a}_j\dr^2$ for $j\neq k$ could provide insights into the sign of the inner product $\ga \vb{x}, \vb{a}_k\dr$. How can we quantify this idea? An interesting quantity would be the

\begin{equation}\label{equality}
    \E\left[\frac{\ga \vb{a}_k, \vb{a}_j \dr^2}{||\vb{a}_k||^2 ||\vb{a}_j||^2}\right]
\end{equation}with $k\neq j$. They should neither be too small because the orthogonal case yields bad results nor too large because we want to explore as many directions as possible. Throughout the paper, we will assume that the $\vb{a}_k$ are Gaussian vectors. in this context, we can establish the following lemma proved in the appendix.

\begin{lemma}
    Let $\vb{a}_1, ..., \vb{a}_T$ be $T$ independent Gaussian vectors in $\R^N$ with a covariance matrix $\vb{C}$ such that $\tau(\vb{C}) = 1$. Let $1\leq k\neq j\leq N$. Then,

\begin{equation}\label{equality}
    \E\left[\frac{\ga \vb{a}_k, \vb{a}_j \dr^2}{||\vb{a}_k||^2 ||\vb{a}_j||^2}\right] = \frac{\tau(\vb{C}^2)}{N}\geq \frac{1}{N}
\end{equation}
with equality when $\vb{C}$ is the identity matrix.
\end{lemma}

Our intuition leads us to believe that when we do not have enough measurements, that is, when $q$ is not close to 0, the second-order moment of $\vb{C}$ will guide the value of $\rho$. In this regime, it will certainly be interesting to choose a large $\tau(\vb{C}^2)$ in order to add significant overlaps between the $\vb{a}_k$.

\subsection{Favoring exploitation over exploration}\label{sectionB}

A natural way to build a simple normalized matrix $\vb{C}$ for a given value of $\tau(\vb{C}^2)$ is to consider $N^* \leq N$ and the $N\times N$ matrix
\begin{equation}\label{matrixlimited}
    \vb{P}_{\alpha} := \vb{O}^T\begin{pmatrix}
    \alpha\vb{I}_{N^*} & \vb{0}\\
    \vb{0} & \vb{0}\\
\end{pmatrix}\vb{O}
\end{equation}with $\alpha = \frac{N}{N^*}\geq 1$ fixed, $\vb{I}_{N^*}$ the $N^*\times N^*$ identity matrix and $\vb{O}$ a Haar distributed random orthogonal matrix. The matrix $\vb{P}_{\alpha}$ is a projector up to a normalization chosen to have $\tau(\vb{P}_{\alpha}) = 1$; note that we have $\tau(\vb{P}_\alpha^2) = \alpha$. We will prove the following theorem.

\begin{theorem}\label{theorem_limited}
Suppose $\vb{a}_1, ..., \vb{a}_T$ are sampled IID from a multivariate Gaussian $\N(\vb{0}, \vb{P}_\alpha)$. $\vb{M}$ still designates the matrix
\begin{eqnarray}
    \vb{M}&=&\frac{1}{T} \sum_{k=1}^T f\left(m_k\right) \vb{a}_k \vb{a}_k^T
\end{eqnarray}with $f:y\mapsto 1 - \frac{1}{y}$ and $m_k:=\ga \vb{a}_k, \vb{x} \dr^2$.
    Then, in the large $N$ and $T$ limit while $q :=\frac{N}{T}$ remains finite,
    \begin{equation}\label{limited}
        \rho(q) := \E\left[\frac{\ga \vb{x}, \vb{y}\dr^2}{||\vb{x}||^2 ||\vb{y}||^2}\right] = \frac{\rho_{cc}\left(\frac{q}{\alpha}\right)}{\alpha}.
    \end{equation}where $\vb{y}$ the eigenvector corresponding to the largest eigenvalue of $\vb{M}$.
\end{theorem}

There is no restriction on making $\alpha$ depend on $q$. An interesting idea would be to consider $\alpha = \frac{q^*}{q}$ with $q^*$ still to be chosen. Let's consider the $N\times N$ matrix
\begin{equation}
    \vb{P}_{q^*} := \vb{O}^T\begin{pmatrix}
    \frac{q}{q^*}\vb{I}_{N^*} & \vb{0}\\
    \vb{0} & \vb{0}\\
\end{pmatrix}\vb{O}
\end{equation}
with $N^* = q^* T$, $\vb{I}_{N^*}$ the $N^*\times N^*$ identity matrix and $\vb{O}$ a Haar distributed random orthogonal matrix. 
\begin{theorem}\label{theorem1}
    Suppose $\vb{a}_1, ..., \vb{a}_T$ are sampled IID from a multivariate Gaussian $\N(0, \vb{P}_{q^*})$. $\vb{M}$ still designates the matrix
\begin{eqnarray}
    \vb{M}&=&\frac{1}{T} \sum_{k=1}^T f\left(m_k\right) \vb{a}_k \vb{a}_k^T
\end{eqnarray}with $f:y\mapsto 1 - \frac{1}{y}$ and $m_k = \ga \vb{a}_k, \vb{x} \dr^2$.
    Then, in the large $N$ and $T$ limit while $q :=\frac{N}{T}$ remains finite,
    \begin{equation}
        \rho(q) := \E\left[\frac{\ga \vb{x}, \vb{y}\dr^2}{||\vb{x}||^2 ||\vb{y}||^2}\right] = \frac{\rho_{cc}(q^*)q^*}{q}.
    \end{equation}where $\vb{y}$ the eigenvector corresponding to the largest eigenvalue of $\vb{M}$.
\end{theorem}

We can plot the function $q\mapsto  \rho_{cc}(q)q$ (Figure \ref{rhoq}).

\begin{figure}[h!]
    \centering
    \includegraphics[scale = 1]{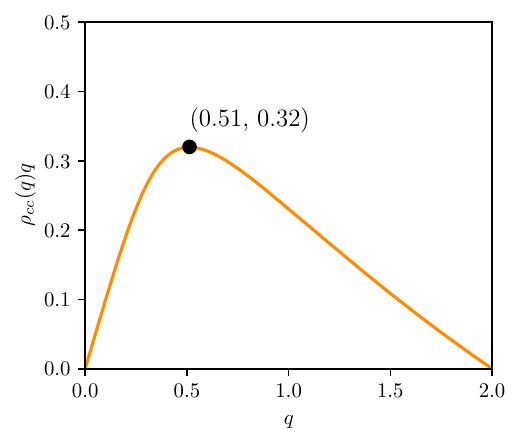}
    \caption{The function $q\mapsto \rho_{cc}(q)q$ with its maximum.}
    \label{rhoq}
\end{figure}

Numerically, we see that the maximum seems to be almost $0.32$ reached for $q = 0.51$. Even with parametric formulas \eqref{formula}, it seems impossible (at least for us) to compute explicitly these numbers. We can thus approximately say that if $q\geq \frac{1}{2} (\approx 0.51)$ it is more interesting to draw the $\vb{a}_k$ vectors in a $\frac{N}{2}$-dimensional space. This method provides a new improved $\rho$ denoted as $\rho_{im}$ which is precisely

\begin{equation}\label{improvedequation}
    \rho_{im}(q) := \frac{\underset{q*\leq q}{\max} \rho_{cc}{(q_c)q_c}}{q}
\end{equation}and plotted in Figure \ref{improved}.

\begin{figure}[h!]
    \centering
    \includegraphics[scale = 1]{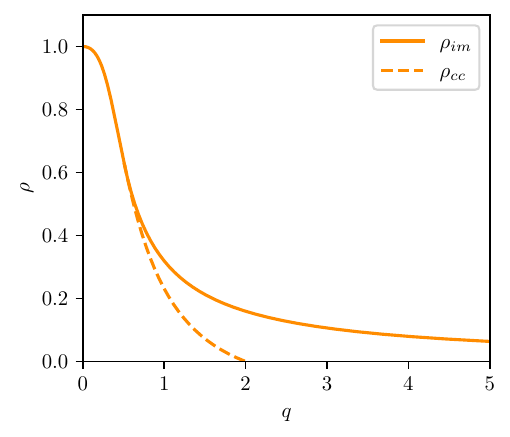}
    \caption{The plot of the $\rho$ improved \eqref{improvedequation}. The dashed line corresponds to the classical case.}
    \label{improved}
\end{figure}

In a way, we were too ambitious, we wanted to explore the entire space while it's more interesting to focus on a limited space and make accurate predictions within it.

\subsection{The theoretical overlap $\rho$ for a Wishart matrix $\vb{C}$}

We have obtained an analytic extension of the formula \eqref{formula} through the following system.

\begin{theorem}\label{generalisation}
We assume that the vectors $\vb{a}_1, ..., \vb{a}_T$ are independent and identically distributed random multivariate Gaussian variables with a covariance matrix $\vb{C}$, randomly drawn from a white-Wishart distribution with $\tau(\vb{C})=1$ and $\tau(\vb{C}^2)=1+p$. $\vb{M}$ still designates the matrix
\begin{eqnarray}
    \vb{M}&=&\frac{1}{T} \sum_{k=1}^T f\left(m_k\right) \vb{a}_k \vb{a}_k^T
\end{eqnarray}with $f:y\mapsto 1 - \frac{1}{y}$ and $m_k = \ga \vb{a}_k, \vb{x} \dr^2$. Then, in the large $N$ and $T$ limit while $q :=\frac{N}{T}$ remains finite, we denote
    \begin{equation}
        \rho(q) := \E\left[\frac{\ga \vb{x}, \vb{y}\dr^2}{||\vb{x}||^2 ||\vb{y}||^2}\right]
    \end{equation}where $\vb{y}$ the eigenvector corresponding to the largest eigenvalue of $\vb{M}$. 
We have 
    \begin{equation}\label{grosysteme}
    \left\{\begin{array}{ccc}
        f(y) &=&1 - \frac{1}{y}\\
        I_1(Z) &=& \int_{-\infty}^{+\infty} \frac{da}{\sqrt{2\pi}} \frac{a^2f^2(a^2)}{Z- f(a^2)}\\
        I_2(Z) &=& \int_{-\infty}^{+\infty} \frac{da}{\sqrt{2\pi}} \frac{f(a^2)}{Z- f(a^2)}\\
        A &=& Z  \\
        B &=& ZI_2(Z)(1+p) - I_1(Z)(1+p)\\
        C &=& -pI_2(Z)I_1(Z) + p Z I_2(Z)^2\\
        q&=& \frac{-B + \sqrt{B^2 - 4 AC}}{2A}\\
        \lambda_1 &=& I_1(Z)\left[1 + p\left(1 + \frac{I_2(Z)}{q}\right)\right]\\
        Z' &=& \frac{1/(1+ p I_2(Z)/q)}{q + I_2(Z) + ZI_2'(Z) + \frac{p\lambda_1 I_2'(Z)}{q(1+p I_2(Z)/q)^2}}\\
        U &=& I_1(Z)\\
        V &=& \frac{p}{\lambda_1}\left[1 + \frac{1}{q}I_2(Z)\right]\\
        U'& =& Z' I_1'(Z)\\
        V' &=& \frac{-p}{\lambda_1^2}\left[1 + \frac{I_2(Z)}{q}\right] + \frac{pZ' I_2'(Z)}{\lambda_1q}\\        
        h' &=&\frac{U'}{1 - UV} + \frac{U(U'V + UV')}{(1-UV)^2}\\
        \rho &=& \frac{1}{1 - h'}.
    \end{array}\right.
\end{equation}
\end{theorem}

$\rho(q)$ can then be plotted parametrically by varying $Z$. The strategy is therefore to calculate all of the above quantities given a value of $Z$ and then to plot $(q, \rho)$. The Figure \ref{wishart} compare theory and numerical simulations when $\vb{C}$ is a Wishart matrix of parameter $p$. 

\begin{figure}[h!]
    \centering
    \includegraphics[scale = 1]{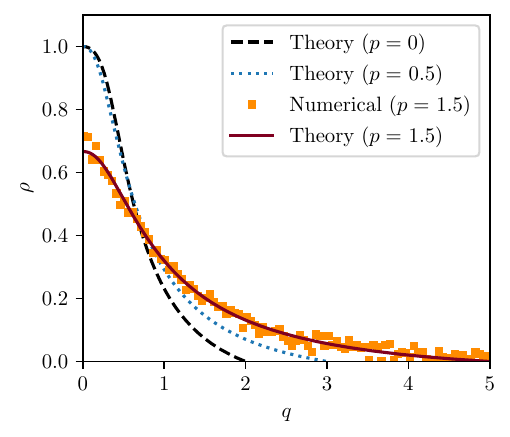}
    \caption{Overlap $\rho = \frac{\ga \vb{x}, \vb{y} \dr^2}{||\vb{x}||^2 ||\vb{y}||^2}$ between the largest eigenvector $\vb{y}$ of $\vb{M}$ and the true signal as a function of $q = \frac{N}{T}$ for the function $f(y) = 1 - 1/y$. Each dot correspond to a single matrix $\vb{M}$ of aspect ratio $q$ and $NT = 10^7$ where the columns were independently drawn from $\mathcal{N}(0, \vb{C})$ where $\vb{C}$ is a Wishart matrix with parameter $p$. The theory curves are given by equations \eqref{grosysteme}. The dashed line corresponds to the classical case when the covariance matrix $\vb{C}$ is the identity studied in the erratum \cite{erratum}.}
    \label{wishart}
\end{figure}

As $q$ approaches 0, which is to say, when the number of measurements is very large, we anticipate that $\mathbf{x}$ will be fully recovered in the explored directions. The spectral method will yield a normalized vector $\mathbf{y}$ that lies in an $N^*$-dimensional space. We denote $\vb{P}$ as the orthogonal projector onto the image space of $\vb{C}$. Given a large $T$ compared to $N$, one might think that 

\begin{equation}
    \frac{\ga \vb{P}\vb{x}, \vb{y}\dr^2}{||\vb{P}\vb{x}||^2} = 1.
\end{equation}

So, the overlap $\rho$ will be

\begin{eqnarray}
    \rho(0) &=& \frac{\ga \vb{x}, \vb{y}\dr^2}{||\vb{x}||^2 ||\vb{y}||^2} = \frac{||\vb{P}\vb{x}||^2}{||\vb{x}||^2} \frac{\ga \vb{P}\vb{x}, \vb{y}\dr^2}{||\vb{P}\vb{x}||^2} \\
    \rho(0)&=& \frac{||\vb{P}\vb{x}||^2}{||\vb{x}||^2}\\
    \rho(0)&=& \frac{\rk(\vb{C})}{N}.\label{q0}
\end{eqnarray}

This is a general remark that does not depend specifically on whether $\vb{C}$ is a Wishart matrix. Let's see now what \eqref{q0} actually says when $\vb{C}$ is a Wishart matrix of a parameter $p$. 

A Wishart matrix of size $N$ with $M$ observations has rank $\min(N,M)$, since $p=N/M$ we have

\begin{equation}
    \rho(0) = \left\{
    \begin{array}{cc}
        1 & \text{ if $p\leq 1$} \\
       \frac{1}{p}  & \text{ if $p\geq 1$} 
    \end{array}
    \right.
\end{equation}.

This is exactly what we observe in Figure \ref{wishart}.

When $q\gtrsim 1$, using a Wishart matrix $\vb{C}$ significantly enhances the performance of the spectral method. We will compute $q_c$, which is the point at which $\rho = 0$, in the next subsection. This value appears to be greater than $2$, which corresponds to the classical case.

\subsection{Around the phase transition}

Let's define the threshold value $q_c$ by the first $q$ such that $\rho(q) = 0$. We know that in the classical case the threshold value is $q_c = 2$. What happens to this value $q_c$ for covariance matrices that deviate from the identity matrix? The following theorems answer this question for the two covariance matrices presented earlier.

\begin{theorem}\label{universal}
    The model where $\vb{C} = \vb{P}_{\alpha}$ presented in \ref{matrixlimited} exhibits a transition at 

    \begin{equation}
        q_c(\vb{P}_{\alpha}) = 2\alpha,
    \end{equation}while the one where $\vb{C}=\vb{W}_p$ is a Wishart matrix with parameter $p$ has a threshold value equal to 

    \begin{equation}
        q_c(\vb{W}_p) = 2(1+p). 
    \end{equation}
    Note that in both cases $q_c(\vb{C}) =2\tau(\vb{C}^2)$.

\end{theorem}

As explained in Section \ref{SectionA}, when we have limited measurements, $\tau(\vb{C}^2)$ drives the value of $\rho$. Thus, in this regime, the $\vb{P}_{\alpha}$ model may be a good candidate to predict the theoretical value of $\rho$ for a normalized covariance matrix with a second-order moment of $\alpha$. In Figure \ref{comparaison}, we test this conjecture in the case of the Wishart model, using the theoretical parametric approach established in the previous section, as well as numerical results for an Inverse Wishart matrix with parameter $p$, whose second moment is also $\tau(\vb{\invW}_p)=1+p$. The behavior of $\rho(q)$ for $q\gtrsim 1$ is similar for all three models sharing the same second moment. We conjecture that $q_c = 2\tau(\vb{C}^2)$ is actually a universal result.
\begin{figure}[h!]
    \centering
    \includegraphics[scale = 1]{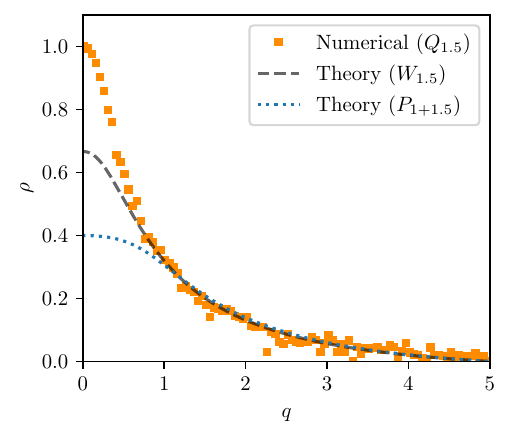}
    \caption{Overlap $\rho = \frac{\ga \vb{x}, \vb{y} \dr^2}{||\vb{x}||^2 ||\vb{y}||^2}$ between the largest eigenvector $\vb{y}$ of $\vb{M}$ and the true signal as a function of $q = \frac{N}{T}$ for the function $f(y) = 1 - 1/y$. Each dot correspond to a single matrix $\vb{M}$ of aspect ratio $q$ and $NT = 10^7$ where the columns were independently drawn from $\mathcal{N}(0, \vb{C})$ where $\vb{C} = \vb{\invW}_{p = 1.5}$ is an inverse-Wishart matrix with parameter $p = 1.5$. The theory curves are given by theorem \ref{theorem_limited} for the model $\vb{P}_{1 + 1.5}$ and theorem \ref{generalisation} if $\vb{C}$ is a Wishart matrix of parameter $p = 1.5$.}
    \label{comparaison}
\end{figure}

\section{Proof of results}
\label{section:proofs}

\subsection{Proof of theorem \ref{orthogonalcase}}

\begin{proof}
We recall that
\begin{equation}
    \vb{M} = \frac{1}{T} \sum_{k=1}^T f\left(\ga \vb{a}_k, x \dr^2\right) \vb{a}_k \vb{a}_k^T.
\end{equation}
Without loss of generality, we can assume the true vector $\vb{x}$ is in the canonical direction $\vb{e}_1$.
The vectors $\vb{a}_k$ form a linearly independent orthogonal family and $\vb{M}$ is thus a diagonal matrix in the basis of the $\vb{a}_k$. Furthermore, the largest eigenvalue of $\vb{M}$ is $f(y_m)$ with $y_m = \ga \vb{a}_m, \vb{e}_1\dr^2 = [\vb{a}_m]_1^2$ as large as possible since $f$ is an increasing function. Therefore, the overlap between the corresponding eigenvector and $\vb{x} = \vb{e}_1$ is 

\begin{eqnarray}
    \rho(q) &=& E\left(\ga \vb{e}_1, \vb{a}_m\dr^2\right)\\
    &=& \E\left([\vb{a}_m]_1^2\right)\\
    &=&\E\left(\underset{1\leq k\leq T}{\max} [\vb{a}_k]_1^2\right)\\
\end{eqnarray}

There are a lot of results about the law of the maximum of random variables with a weak correlation. But, here, the variables are supposed to be orthogonal and so strongly correlated. Since $\vb{O}$ and $\vb{O}^T$ share the same distribution, the distribution of the maximum of the $T$ first squared elements in the first row of $\vb{O}$ is identical to that of the maximum of the $T$ first squared elements in the first column of $\vb{O}$. Moreover, the first column of a randomly generated orthogonal matrix follows a uniform distribution on the $(N-1)-$sphere. Finally, we get

\begin{equation}
    \rho = \E\left[\underset{1\leq k\leq T}{\max} \frac{\vb{b}_k^2}{||\vb{b}||^2}\right] = \frac{1}{N} \E \left[\underset{1\leq k\leq T}{\max} [\vb{b}_k]_1^2\right].
\end{equation}where $(\vb{b}_k)_{1\leq k\leq T}$ are IID Gaussian variables. Indeed, when $N$ gets large, $||\vb{b}||^2$ is very close to its expected value which is $N$. In addition, it is a well-known fact that

\begin{equation}
    \underset{1\leq k\leq T}{\max} [\vb{b}_k]_1 \underset{T\to \infty}{\sim} \sqrt{2 \log(T)}
\end{equation}
with small fluctuations.

Hence, we have

\begin{equation}
    \rho = \frac{2 \log(T)}{N}
\end{equation}
which is equivalent to

\begin{equation}
    \frac{2\log(N)}{N}
\end{equation}
when $q = \frac{N}{T}$ remains finite while $N, T\longrightarrow +\infty$.

So, 

\begin{equation}
    \rho(q) \underset{N\to \infty}{\longrightarrow} 0.
\end{equation}
and hence the result.
\end{proof}

\subsection{Proof of theorem \ref{theorem1}}
Theorems \ref{theorem_limited} and \ref{theorem1} are equivalent and we will prove the second one.

\begin{proof}    
Notice that $\vb{C} = \frac{q}{q^*}\vb{P}$ where $\vb{P}$ is the normalized projection on the $N^*$-dimensional space spanned by the $N^*$ first columns of $\vb{O}$. We can assume that 

\begin{equation}
    \vb{H} = \vb{C}^{1/2} \vb{H}_0
\end{equation}
where all the coefficients of $\vb{H}_0$ are IID normalized Gaussian random variables. We denote $\vb{b}_k$ the columns of $\vb{H}_0$.
Thus, we have

\begin{equation}
    \vb{M} = \frac{1}{T}\sum_{k=1}^T f\left(\ga \vb{C}^{1/2} \vb{b}_k, x\dr^2\right) \vb{C}^{1/2}\vb{b}_k \vb{b}_k^T \vb{C}^{1/2}
\end{equation}

\begin{multline}
    =\frac{1}{T}\sum_{k=1}^T f\left(\ga \vb{O}^T\vb{D}^{1/2} \vb{O} \vb{b}_k, \vb{x}\dr^2\right) \\\vb{O}^T \vb{D}^{1/2}\vb{O}\vb{b}_k \vb{b}_k^T \vb{O}^T \vb{D}^{1/2} \vb{O}
\end{multline}

\begin{multline}
    =\frac{1}{T}\sum_{k=1}^T f\left(\ga \vb{D}^{1/2} (\vb{O} \vb{b}_k), \vb{O}\vb{x}\dr^2\right)\\ \vb{O}^T \vb{D}^{1/2}(\vb{O}\vb{b}_k) (\vb{O}\vb{b}_k)^T \vb{D}^{1/2} \vb{O}
\end{multline}

$\vb{O}\vb{b}_k$ and $\vb{b}_k$ follow the same distribution and replacing $\vb{O}\vb{b}_k$ by $\vb{b}_k$ will not change the final overlap $\rho$. We can therefore assume

\begin{equation}
     \vb{M} = \frac{1}{T}\sum_{k=1}^T f\left(\ga \vb{D}^{1/2} \vb{b}_k, \vb{O}\vb{x}\dr^2\right) \vb{O}^T (\vb{D}^{1/2}\vb{b}_k) (\vb{D}^{1/2}\vb{b}_k)^T  \vb{O}.
\end{equation}

One can check that if $\vb{y}$ is the eigenvector corresponding to the largest eigenvalue of a matrix $\vb{A}$, then $\vb{O}^T \vb{y}$ is the eigenvector corresponding to the largest eigenvalue of the matrix $\vb{O}^T \vb{A}\vb{O}$.

We denote $\vb{y}$ the eigenvector corresponding to the largest eigenvalue of $\vb{M}$ and we study the quantity $\ga \vb{x}, \vb{y}\dr^2$. We can rewrite it as $\ga \vb{O}\vb{x}, \vb{O}\vb{y}\dr^2$. So we can assume that

\begin{equation}\label{eq2}
     \vb{M} = \frac{1}{T}\sum_{k=1}^T f\left(\ga \vb{D}^{1/2} \vb{b}_k, \vb{x}\dr^2\right) (\vb{D}^{1/2}\vb{b}_k) (\vb{D}^{1/2}\vb{b}_k)^T.
\end{equation}
and we study the overlap between the largest eigenvector of such a matrix and $\vb{x}$.

Let's take $\vb{D} = \frac{q}{q^*}\vb{P}$ in order to obtain

    \begin{equation}\label{eq2}
     \vb{M} \propto \frac{1}{T}\sum_{k=1}^T f\left(\ga \vb
     P\vb{b}_k, \frac{\vb{P}\vb{x}}{||\vb
     P\vb{x}||}\dr^2\right) \vb{P}\vb{b}_k (\vb{P}\vb{b}_k)^T
\end{equation}
since $||\vb{P}\vb{x}||^2 = \frac{q^*}{q}$.

This matrix will give a vector $\vb{y}$ such that

\begin{equation}
    \E\left[\ga \frac{\vb{P}\vb{x}}{||\vb{P}\vb{x}||}, \vb{y}\dr^2\right] = \rho_{cc}(q^*)
\end{equation}
and therefore

\begin{equation}
    \E\left[\ga \vb{x}, \vb{y}\dr^2\right] = \frac{q^* \rho_{cc}(q^*)}{q}.
\end{equation}
    
\end{proof}

\subsection{Proof of theorem \ref{generalisation}}

\begin{proof}

The problem is invariant by rotation. Without loss of generality, we can assume true vector $\vb{x}$ is in the canonical direction $\vb{e}_1$. So, $m_k$ are just $[\vb{a}_k]_1^2$ with $k=1, ..., T$. The core concept of the proof is to consider the conditional probabilities given the quantities $[\vb{a}_k]_1$ which play a role different from the other coordinates since they are in the argument of $f$ in the expression of $\vb{M}$.

We write $\vb{C}$ as
\begin{equation}
    \begin{pmatrix}
        \vb{C}_{11} & \vb{b}^T\\
        \vb{b} & \vb{\hat{C}}
    \end{pmatrix}
\end{equation}
with $\vb{C}_{11}\in \R, \vb{b}\in \R^{N-1}$ and $\vb{\hat{C}}$ is an $(N-1)\times (N-1)$ matrix.

Let $1\leq k\leq T$. Given $[\vb{a}_k]_1$, the vector $\vb{\hat{a}}_k\in \R^{N-1}$ which is $\vb{a}_k$ without the first coordinate, is also a Gaussian vector defined by

\begin{equation}
    \left\{
    \begin{array}{lll}
        \E(\vb{\hat{a}}_k | [\vb{a}_k]_1) &=& \frac{\vb{b}}{\vb{C}_{11}} [\vb{a}_k]_1\\
         \text{Cov}(\vb{\hat{a}}_k | [\vb{a}_k]_1) &=& \vb{\Sigma} :=  \hat{C} - \frac{\vb{b} \vb{b}^T}{C_{1 1}}\\
    \end{array}
    \right.
\end{equation}
where $\Sigma$ is the $(N-1)\times (N-1)$ matrix. This is an elementary result using the Schur Complement formula.

We can write $\vb{\hat{a}}_k$ in a most satisfying way:

\begin{equation}
    \vb{\hat{a}}_k = \vb{v} [\vb{a}_k]_1 + \vb{e}_k
\end{equation}
with 
\begin{equation}
    \vb{v} = \frac{\vb{b}}{\vb{C}_{11}}
\end{equation}
and $\vb{e}_k$ a $(N-1)$ vector drawn from $\mathcal{N}(0, \vb{\Sigma})$.

Let's put all the $\vb{e}_k$ in a $(N-1)\times T$ matrix called $\vb{E}$:
\begin{equation}
    \vb{E} := (\vb{e}_k)_{1\leq k\leq T}.
\end{equation}

We now denote the $[\vb{a}_k]_1$ by $d_k$. We insist that the above operation allows us to consider the $d_k$ independent from the other variables. We recall that

\begin{equation}
    \vb{M} = \frac{1}{T} \sum_{k=1}^T f(d_k^2) \vb{a}_k \vb{a}_k^T.
\end{equation}
We use the same block decomposition as $C$'s one, 

\begin{equation}
    \begin{pmatrix}
        \vb{M}_{11} & \vb{M}_{21}^T\\
        \vb{M}_{21} & \vb{M}_{22}
    \end{pmatrix}
\end{equation}
where $\vb{M}_{21}\in \R^{N-1}$ and $\vb{M}_{22}$ is a $(N-1)\times (N-1)$ matrix.

We easily see that

\begin{equation}\label{3M2}
    \left\{
    \begin{array}{l}
        \vb{M}_{11} = \frac{1}{T} \sum_{k=1}^T f(d_k^2) d_k^2\\
        \vb{M}_{21} = \frac{1}{T} \sum_{k=1}^T f(d_k^2) d_k (d_k \vb{v} + \vb{e}_k)\\
        \vb{M}_{22} = \frac{1}{T} \sum_{k=1}^T f(d_k^2) (d_k \vb{v} + \vb{e}_k) (d_k \vb{v} + \vb{e}_k)^T.
    \end{array}
    \right.
\end{equation}

Let's recall that $f$ is defined by $y\mapsto 1 - \frac{1}{y}$, and so

\begin{equation}
    \vb{M}_{11} = \frac{1}{T} \sum_{k=1}^T (d_k^2 -1).
\end{equation}

The $d_k$ are independent, so the law of large numbers tells us that

\begin{equation}
    \vb{M}_{11} \underset{T\to +\infty}{\approx} \E(d_k^2) - 1 = \vb{C}_{11} - 1.
\end{equation}

Furthermore, when $T$ gets large, $\vb{C}_{11}$ is very close to its expected value which is $\tau(\vb{C}) = 1$. \eqref{3M2} becomes

\begin{equation}\label{3M}
    \left\{
    \begin{array}{l}
        \vb{M}_{11} = 0\\
        \vb{M}_{21} = \vb{u}\\
        \vb{M}_{22} = \vb{A} + \vb{u} \vb{v}^T + \vb{v}\vb{u}^T
    \end{array}\right.
\end{equation}
with

\begin{equation}
    \vb{u} := \frac{1}{T} \sum_{k=1}^T f(d_k^2) d_k \vb{e}_k.
\end{equation}
and 

\begin{equation}
    \vb{A} := \frac{1}{T} \sum_{k=1}^T f(d_k^2) \vb{e}_k \vb{e}_k^T.
\end{equation}

Now, we look for the zeros of the Stieljes transform $g_{\vb{M}}(z) = \tau((z\vb{1} - \vb{M})^{-1})$ in order to find an outlier. Using  Schur complement and Sherman-Morisson formulas, we get

\begin{equation}
    N \g_{\vb{M}}(z) = \Tr(\vb{G}_{22}(z)) + \frac{1 + \Tr(\vb{G}_{22}(z)\vb{M}_{21}\vb{M_{12}} \vb{G}_{22}(z))}{z - \vb{M}_{11} - \vb{M_{12}}\vb{G}_{22}(z)\vb{M}_{21}}
\end{equation}
where $\vb{G}_{22}(z)$ is the matrix resolvent of $\vb{M}_{22}$. As $\vb{M}_{22}$ is a rotationally invariant matrix, we expect it to have a continuous spectrum with an edge $\lambda_{+}$. $\lambda$ is an outlier of $\vb{M}$ if it is a pole of $\g_{\vb{M}}(z)$. In other words, $\lambda_1$ satisfied the fixed point equation:
\begin{equation}
    \lambda_1 = h(\lambda_1)
\end{equation}
where

\begin{equation}
    h(z) = \vb{M}_{11} - \vb{M_{12}}\vb{G}_{22}(z)\vb{M}_{21} = \vb{u}^T \vb{G}_{22}(z) \vb{u}
\end{equation}which is assumed to be self-averaging.

The overlap $\rho$ between the corresponding eigenvector $\vb{y}$ and $\vb{x}$ is given by the residues

\begin{equation}
    \rho = \frac{\ga \vb{y}, \vb{x}\dr^2}{||\vb{x}||^2} = \underset{z\to \lambda_1}{\text{lim}} \frac{z-\lambda_1}{z - h(z)} = \frac{1}{1 - h'(\lambda_1)}
\end{equation}
using l'Hospital's rule.

Since $\vb{M}_{22}$ is a rank-$2$ perturbation of $\vb{A}$, we can apply Sherman-Morrison formula twice to express $\vb{G}_{22}$ given $\vb{G}_{\vb{A}}$ the resolvent matrix of $\vb{A}$. Eventually, we get

\begin{multline}
    h(z) = \vb{u}^T\vb{G}_{\vb{A}}\vb{u} - \vb{u}^T\vb{G}_{\vb{A}} \vb{v} \vb{u}^T \vb{G}_{\vb{A}}\vb{u} \\- \frac{\vb{u}^T(\vb{G}_{\vb{A}} - \vb{G}_{\vb{A}} \vb{v} \vb{u}^T \vb{G}_{\vb{A}}) \vb{u} \vb{v}^T (\vb{G}_{\vb{A}} - \vb{G}_{\vb{A}} \vb{v} \vb{u}^T \vb{G}_{\vb{A}})\vb{u}}{1 - \vb{v}^T \vb{G}_{\vb{A}} \vb{v} \vb{u}^T \vb{G}_{\vb{A}}\vb{u}}.
\end{multline}
As $\vb{u}$ and $\vb{v}$ are zero-mean random variables, we expect $\vb{u}^T \vb{G}_{\vb{A}} \vb{v}$ or $\vb{v}^T \vb{G}_{\vb{A}} \vb{u}$ to vanish. By using such an approximation, we get

\begin{equation}\label{h}
    h(z) = U + \frac{U^2 V}{1 -UV} = \frac{U}{1-UV}
\end{equation}
with $U$ and $V$ being self-averaging quantities defined by $U := \vb{u}^T \vb{G}_{\vb{A}}\vb{u}$ et $V := \vb{v}^T \vb{G}_{\vb{A}}\vb{v}$. 

Now, we should proceed to compute $U$ and $V$.

We have

\begin{align}
    U&=\vb{u}^T \vb{G}_{\vb{A}}\vb{u}\\
    &=\frac{1}{T^2}\sum_{k, l=1}^T f(d_k^2)f(d_l^2) d_k d_l \sum_{i, j =1}^{N-1} [\vb{a}_k]_i [\vb{G}_{\vb{A}}(z)]_{i, j} [\vb{a}_l]_j\\
    &=\frac{1}{T^2}\sum_{k=1}^T f^2(d_k^2)d_k^2 \sum_{i, j =1}^{N-1} [\vb{a}_k]_i [\vb{G}_{\vb{A}}(z)]_{i, j} [\vb{a}_l]_j
\end{align}
since $U$ is self-averaging and the two random variables $\vb{A}$ and $\vb{D}_0$ are independent.

This leads to the equation 
\begin{equation}
    U = q\tau\left(\frac{\vb{E} \vb{D} \vb{E}^T}{T} \left(z\vb{1} - \frac{\vb{E} \vb{D}_0 \vb{E}^T}{T}\right)^{-1} \right)
\end{equation}
where $[\vb{D}]_{kl} =f^2(d_k^2)d_k^2 \delta_{kl}$.

We use the same manipulations as in the erratum \cite{erratum}:

\begin{align}
    U &=q\tau\left(\frac{\vb{E} \vb{D} \vb{E}^T}{T} \left(z\vb{1} - \frac{\vb{E} \vb{D}_0 \vb{E}^T}{T}\right)^{-1} \right)\\
    &= \frac{1}{T} \Tr\left[\frac{\vb{E} \vb{D} \vb{E}^T}{T} \sum_{n=0}^{+\infty} \frac{1}{z^{n+1}} \left(\frac{\vb{E} \vb{D}_0 \vb{E}^T}{T}\right)^n \right]\\
    &= \frac{1}{T} \Tr\left[\vb{D} \vb{D_0}^{-1/2} \sum_{n=0}^{+\infty} \frac{1}{z^{n+1}}\left(\vb{D_0}^{1/2} \frac{\vb{E}^T \vb{E}}{T}\vb{D_0}^{1/2}\right)^{n+1}\vb{D_0}^{-1/2}\right]\\
   & =\frac{1}{T} \Tr\left[\vb{D}\vb{D_0}^{-1/2} \vb{T}_{\vb{W}}(z)\vb{D_0}^{-1/2}\right]
\end{align}
where $\vb{T}_{\vb{W}}$ is the $\vb{T}$-matrix of 

\begin{equation}
    \vb{W}:= \vb{D_0}^{1/2} \frac{\vb{E}^T \vb{E}}{T}\vb{D_0}^{1/2} = \vb{X}^{1/2} \vb{Y}\vb{X}^{1/2}
\end{equation}
with $\vb{X} = q\vb{D_0}$ and $\vb{Y} = \frac{\vb{E}^T\vb{E}}{N}$.

Since $\vb{X}$ and $\vb{Y}$ are free, we can use in the large $N$ limit the following subordination relation, as recalled in the Appendix \ref{app:appendixRMT}, which contains the $S$-transform of $\vb{Y}$:

\begin{eqnarray}
    \E[\vb{T}_{\vb{W}}(z)]_{\vb{Y}} &=& \vb{X}\left[z S_{\vb{Y}}(\mathfrak{t}_{\vb{W}}(z))\vb{1} - \vb{X}\right]^{-1}
\end{eqnarray}

with

\begin{equation}
    S_{\vb{Y}}(t) = \frac{S_{\vb{C}}\left(\frac{t}{q}\right)}{1 + \frac{t}{q}}
\end{equation}

according to Lemma \ref{lemme2} in Appendix \ref{app:lemma}.

Each $d_k$ is a normal variable with mean of $0$ and a standard deviation of $\tau(\vb{C}) = 1$. Therefore we can now put everything together to obtain

\begin{equation}\label{U}
    U = I_1(Z)
\end{equation}
with

\begin{equation}
    I_1(Z) = \int_{-\infty}^{+\infty} \frac{x^2 f^2(x^2)}{Z - f(x^2)}\frac{dx}{\sqrt{2\pi}},
\end{equation}

\begin{equation}\label{Zz}
    Z = z\frac{S_{\vb{C}}\left(\frac{I_2(Z)}{q}\right)}{q + I_2(Z)}
\end{equation}
and

\begin{equation}
    I_2(Z) = \int_{-\infty}^{+\infty} \frac{f(x^2)}{Z - f(x^2)}\frac{dx}{\sqrt{2\pi}}.
\end{equation}

Previously, we have assumed that $\vb{C}_{11} = 1$. So,

\begin{equation}\label{V2}
    V = \tau\left[\vb{b}\vb{b}^T \left(z\vb{1} - \frac{\vb{E}\vb{D_0}\vb{E}^T}{T}\right)^{-1}\right].
\end{equation}

So far, our calculations have been quite general and applicable to any rotation-invariant covariance matrix. Now, we will truly make use of the fact that $\vb{C}$ is a Wishart matrix, as we do not know how to calculate such a quantity $V$ \eqref{V2} in a general context.

Lemma \ref{lemme1} in Appendix \ref{app:lemma} yields

\begin{equation}
    V = p\tau\left[\left(z\vb{1} - \frac{\vb{E}\vb{D_0}\vb{E}^T}{T}\right)^{-1}\right]
\end{equation}with $p = \tau(\vb{C}^2) -1$.

We now make $\vb{T}_{\vb{W}}$ appear in such a way as to be able to use the same previous manipulations:

\begin{align}
     V &= p \frac{\Tr}{N}\left[\left(z\vb{1} - \frac{\vb{E}\vb{D}_0\vb{E}^T}{T}\right)^{-1}\right]\\
    &= \frac{p}{z} \frac{\Tr}{N}\left[\left(z\vb{1} - \frac{\vb{E}\vb{D}_0\vb{E}^T}{T} + \frac{\vb{E}\vb{D}_0\vb{E}^T}{T}\right)\left(z\vb{1} - \frac{\vb{E}\vb{D}_0\vb{E}^T}{T}\right)^{-1}\right]\\
    &= \frac{p}{z}\left(1 + \frac{\Tr}{N}\left[\frac{\vb{E}\vb{D}_0\vb{E}^T}{T}\left(z1 - \frac{\vb{E}\vb{D}_0\vb{E}^T}{T}\right)^{-1}\right]\right)\\
    &= \frac{p}{z}\left(1 + \frac{\Tr}{N}\vb{T}_{\vb{W}}(z)\right)\\
    &= \frac{p}{z}\left(1 + \frac{\mathfrak{t}_{\vb{W}}(z)}{q}\right)\\
    &= \frac{p}{z}\left(1 + \frac{I_2(Z)}{q}\right)\label{V}.
\end{align}

The variable $Z$ seems to be a useful new variable, we would like to express $\lambda_1$ in terms of $Z$. We know that $h(\lambda_1) =\lambda_1$, so we have

\begin{equation}
    \lambda_1= \frac{U}{1 - U V}
\end{equation}by using \ref{h}.

The use of \eqref{U} and \eqref{V} at $z = \lambda_1$ yields 

\begin{eqnarray}
    \lambda_1&=& \frac{U}{1 - U V}\\
    \lambda_1 (1 - UV) &=& U\\
    \lambda_1\left(1 - I_1(Z) \frac{p}{\lambda_1} \left(1 + \frac{I_2(Z)}{q}\right)\right)& = &I_1(Z)   
\end{eqnarray}
and so,

\begin{equation}\label{equ}
    \lambda_1 = I_1(Z)\left[1 + p\left(1 + \frac{I_2(Z)}{q}\right)\right].
\end{equation}

Now, we would like to find an expression of $q$ as a function of $Z$. We know

\begin{equation}
    Z = \lambda_1 \frac{S_{\vb{C}}\left(\frac{I_2(Z)}{q}\right)}{q + I_2(Z)}
\end{equation} which gives us the equation

\begin{equation}\label{Z}
    Z = I_1(Z)\left(1 + p\left(1 + \frac{I_2(Z)}{q}\right)\right) \frac{S_{\vb{C}}\left(\frac{I_2(Z)}{q}\right)}{q + I_2(Z)}
\end{equation}by using \eqref{equ}.

We recall that the $S$-transform of a Wishart matrix $\vb{W}_p$ with $p$ as the parameter is given by

\begin{equation}
S_{\vb{W}_p}(t) = \frac{1}{1 + pt}
\end{equation}(see a proof in \cite{potters2020first}). 

Eventually, \ref{Z} gives a quadratic equation for $q$:
\begin{equation}
    A q^2 + B q + C = 0
\end{equation}
with

\begin{equation}\label{system}
    \left\{\begin{array}{ccc}
        A &=& Z  \\
        B &=& ZI_2(Z)(1+p) - I_1(Z)(1+p)\\
        C &=& -pI_2(Z)I_1(Z) + p Z I_2(Z)^2
    \end{array}\right..
\end{equation}

Then, let's find out an expression of $Z'$ in terms of $q$ and $Z$. By using \ref{Zz}, we start from 

\begin{equation}\label{q}
    Z(q + I_2(Z)) = zS_{\vb{C}}\left(\frac{I_2(Z)}{q}\right)
\end{equation}which can be differentiated with respect to $z$:

\begin{multline}
    Z'(q + I_2(Z)) + Z Z'I_2'(Z) = \\S_{\vb{C}}\left(\frac{I_2(Z)}{q}\right) + zZ' \frac{I_2'(Z)}{q} S_{\vb{C}}'\left(\frac{I_2(Z)}{q}\right)
\end{multline}
and as result, we get

\begin{equation}\label{ZZ}
    Z' = \frac{S_{\vb{C}}\left(\frac{I_2(Z)}{q}\right)}{q + I_2(Z) + Z I_2'(Z) - \lambda_1  \frac{I_2'(Z)}{q} S_{\vb{C}}'\left(\frac{I_2(Z)}{q}\right)}
\end{equation}
at the point $z = \lambda_1$.

Since, $S_{\vb{C}}(t) = \frac{1}{1 + pt}$, we get

\begin{equation}
    Z' = \frac{1/(1+ p I_2(Z)/q)}{q + I_2(Z) + ZI_2'(Z) + \frac{p\lambda_1 I_2'(Z)}{q(1+p I_2(Z)/q)^2}}.
\end{equation}
By differentiating \eqref{h}, we have

\begin{equation}
    h'(\lambda_1) = \frac{U'}{1 - UV} + \frac{U(U'V + UV')}{(1-UV)^2}
\end{equation}
with

\begin{equation}
    U' = Z' I_1'(Z)
\end{equation}
and

\begin{equation}
    V' = \frac{-p}{\lambda_1^2}\left[1 + \frac{I_2(Z)}{q}\right] + \frac{pZ' I_2'(Z)}{\lambda_1q}.
\end{equation}.

We have obtained all the equations of the system \eqref{grosysteme}.

\end{proof}
\subsection{Proof of theorem \ref{universal}}

By using theorem \ref{theorem_limited}, it is easy to see that

\begin{equation}
    q_c(\vb{P}_{\alpha}) = q_{cc} \alpha = 2\alpha
\end{equation}where $q_{cc} = 2$ is the threshold value in the classical case.

We will now compute $q_c$ in the Wishart case using the system \ref{grosysteme}. The integrals $I_1$ and $I_2$ are not well-defined for $Z\leq Z^* = 1$. What is the value of $q$ corresponding to $Z^*=1$? By using Lemma \ref{integrals}, we can find that 

\begin{equation}
    \left\{
    \begin{array}{ccc}
        A(Z = 1) &=&1  \\
        B(Z=1) & = &- 2(1+p)\\
        C(Z = 1) & = & 0
    \end{array}\right.
\end{equation}

and so we get

\begin{equation}\label{qZ}
    q(Z = 1) = \frac{-B}{A} = 2(1+p).
\end{equation}

Furthermore, using Lemma \ref{integrals}, we can easily see that

\begin{equation}\label{lZ}
    \lambda_1(Z = 1) = 2(1+p)
\end{equation}

If one looks the system \ref{grosysteme}, we see that the only way that $\rho = 0$ is that $Z' = -\infty$ and so

\begin{equation}\label{equationnulle}
    D(Z):=q + I_2(Z) + ZI_2'(Z) + \frac{p\lambda_1 I_2'(Z)}{q(1+p I_2(Z)/q)^2} = 0.
\end{equation}

By using \eqref{qZ}, \eqref{lZ} and Lemma \ref{integrals} in the appendix, we get

\begin{equation}
    D(Z = 1) = 2(1+p) - 2 + \frac{-4(1+p)p}{2(1+p)} = 0.
\end{equation}

Hence the result because

\begin{equation}
    \left\{\begin{array}{ll}
        q(Z = 1) &= 2(1+p)\\
        \rho(Z = 1) & = 0
    \end{array}\right..
\end{equation}

\appendices
\section{REMINDER ON TRANSFORMS IN RMT}
\label{app:appendixRMT}

The analog of the expectation value in the world of non-commutative random matrices is the normalized trace operator $\tau$ defined as

\begin{equation}
    \tau(\vb{A}) = \frac{1}{N} \E[\Tr(\vb{A})]
\end{equation}which remains finite even as $N$ goes to infinity.

We give a brief recapitulation of various transforms that hold significance in the analysis of eigenvalues statistics in Random Matrix Theory, as they are intricately linked with free probability theory (see e.g. \cite{speicher2011free}, \cite{burda2013free} or \cite{potters2020first}). A standard way to approach the distribution of eigenvalues of a random matrix $\vb{M}$ is to introduce its resolvent matrix $\vb{G}_{\vb{M}}$, defined as

\begin{equation}
    \vb{G}_{\vb{M}}(z) := (z\vb{1} - \vb{M})^{-1}
\end{equation}
and its normalized trace, called the Stieltjes transform of $\vb{M}$, defined as 

\begin{eqnarray}
    \mathfrak{g}_{\vb{M}}(z) &:=& \frac{1}{N} \Tr(\vb{G}_{\vb{M}}(z))\\
    &=& \frac{1}{N}\sum_{k=1}^N \frac{1}{z - \lambda_k}\\
    &\underset{N\to \infty}{\sim}& \int \frac{\rho_{\vb{M}}(\lambda)}{z-\lambda}d\lambda.
\end{eqnarray}

In both expressions, $z$ is in the complex plane but outside the real axis to avoid the poles of $\vb{G}_{\vb{M}}$ which are the eigenvalues of $\vb{M}$. If $z$ is chosen to be not close to the real axis and it turns out that $\mathfrak{g}_{\vb{M}}(z)$ is self-averaging in the large $N$ limit and its value is independent of the specific realization of $\vb{M}$. 

The $\vb{T}$-transform of $\vb{M}$ is defined by

\begin{equation}
    \vb{T}_{\vb{M}}(z) := \vb{M}\left(z\vb{1} - \vb{M}\right)^{-1} = z \vb{G}_{\vb{M}}(z) - 1
\end{equation}

and the $\mathfrak{t}_{\vb{M}}$ is

\begin{equation}
    \mathfrak{t}_{\vb{M}}(z) := \tau\left[\vb{M}\left(z\vb{1} - \vb{M}\right)^{-1}\right] = z\mathfrak{g}_{\vb{M}}(z) - 1.
\end{equation}

The $S$-transform of $\vb{M}$ is then defined as

\begin{equation}\label{definitionS}
    S_{\vb{M}}(t) := \frac{t+1}{t \mathfrak{t}_{\vb{M}}^{-1}(t)}
\end{equation}where $\mathfrak{t}_{\vb{M}}^{-1}(t)$ is the functional inverse of the $\mathfrak{t}$-transform.

Let $\vb{X}$ and $\vb{Y}$ be two random free symmetric matrices and $\vb{W} = \sqrt{\vb{X}}\vb{Y}\sqrt{\vb{X}}$. The result, first obtained in \cite{voiculescu1992free}, reads:

\begin{equation}
    S_{\vb{W}}(t) = S_{\vb{X}}(t) S_{\vb{Y}}(t).
\end{equation}

Moreover, in \cite{bun2016rotational}, a Replica analysis leads to the following subordination relation:

\begin{eqnarray}
    \E[\vb{T}_{\vb{W}}(z)]_{\vb{Y}} &=& \vb{T}_{\vb{X}}[z S_{\vb{Y}}(\mathfrak{t}_{\vb{W}}(z))]\\
    \E[\vb{W}(z\vb{1} - \vb{W})^{-1}]_{\vb{Y}} &=& \vb{X}[z S_{\vb{Y}}(\mathfrak{t}_{\vb{W}}(z))\vb{1} - \vb{X}]^{-1}.
\end{eqnarray}
where $\vb{X}$ and $\vb{Y}$ are two independent random matrices and $\vb{Y}$ is rotationally invariant. It is readily apparent that we also have the following equation:

\begin{equation}
    \E[\vb{\mathfrak{t}}_{\vb{W}}(z)]_{\vb{Y}} = \vb{\mathfrak{t}}_{\vb{X}}(z S_{\vb{Y}}(\mathfrak{t}_{\vb{W}}(z))).
\end{equation}

If $\vb{H}$ designates a $N\times T$ matrix filled with IID standard Gaussian random numbers, we say that

\begin{equation}
    \vb{W} = \frac{1}{T} \vb{H}\vb{H}^T
\end{equation}is a (white) Wishart matrix of parameter $q = \frac{N}{T}$.

In addition, the inverse-Wishart matrix $\vb{\invW}_p$ is defined as 
\begin{equation}
    \vb{\invW}_p = (1 -q) \vb{W}_q^{-1}
\end{equation}where $\vb{W}_q$ is a Wishart matrix with parameter $q$ such that $p = \frac{q}{1-q}$. 

Note that both $\vb{W}_p$ and $\vb{\invW}_p$ are normalized in the sense that $\tau(\vb{W}_p) = \tau(\vb{\invW}_p) = 1$ and furthermore they have the same second-order moment which is $1+p$.

We recall that the $S$-transform of a Wishart matrix $\vb{W}_p$ with $p$ as the parameter is given by

\begin{equation}
S_{\vb{W}_p}(t) = \frac{1}{1 + pt}
\end{equation}
and that of an inverse Wishart matrix $\vb{\invW}_p$ with parameter $p$ is given by

\begin{equation}
S_{\vb{\invW}_p}(t) = 1 - pt
\end{equation}(see a proof in \cite{potters2020first}). 

In many instances, the eigenvalue spectrum of large random matrices is confined to a single finite-size interval, denoted as $[\lambda_{-}, \lambda_{+}]$. This confinement is observed in various random matrix ensembles, including Wigner matrices, where properly normalized eigenvalues fall within the range $\lambda_{-} = -2$ and $\lambda_{+} = +2$, exhibiting a semicircular distribution between these two edges. However, as $N$ becomes large but finite, it is expected that the maximum eigenvalue $\lambda^{+}$ will surpass the upper edge $\lambda_{+}$. In many random matrix ensembles, the fluctuations of $\lambda_{\text{max}}$ around $\lambda_{+}$ follow Tracy-Widom statistics and are of order $N^{-2/3}$. In our specific problem, $\lambda_{\text{max}}$ will experience larger fluctuations beyond the $N^{-2/3}$ scale. Therefore, we will employ a different approach to study these outliers.

\section{Some useful lemmas}
\label{app:lemma}
\begin{lemma}
    Let $\vb{a}_1, ..., \vb{a}_T$ be $T$ independent Gaussian vectors in $\R^N$ with a covariance matrix $\vb{C}$ such that $\tau(\vb{C}) = 1$. Let $1\leq k\neq j\leq N$. Then,

\begin{equation}\label{equality}
    \E\left[\frac{\ga \vb{a}_k, \vb{a}_j \dr^2}{||\vb{a}_k||^2 ||\vb{a}_j||^2}\right] = \frac{\tau(\vb{C}^2)}{N}\geq \frac{1}{N}
\end{equation}
with equality when $\vb{C}$ is the identity matrix.
\end{lemma}

\begin{proof}

In the high-dimensional limit, $||\vb{a}_k||^2, ||\vb{a}_j||^2$ are very close to $N$ because $\tau(\vb{C})=1$. So, we get

\begin{equation}
    \E\left[\frac{\ga \vb{a}_k, \vb{a}_j \dr^2}{||\vb{a}_k||^2 ||\vb{a}_j||^2}\right] = \frac{1}{N^2} \E\left[\ga \vb{a}_k, \vb{a}_j \dr^2\right].
\end{equation}

The Independence of the $\vb{a}_k$ yields

\begin{multline}
    \E\left[\ga \vb{a}_k, \vb{a}_j \dr^2\right]= \\
    \sum_{1\leq v, w\leq N}  \E\left([\vb{a}_k]_v [\vb{a}_k]_w\right)\E\left([\vb{a}_j]_v [\vb{a}_j]_w\right)
\end{multline}

\begin{eqnarray}
    \E\left[\ga \vb{a}_k, \vb{a}_j \dr^2\right]&=& \sum_{1\leq v, w\leq N} [\vb{C}]_{v, w}[\vb{C}]_{v, w}\\
    &=& \sum_{v=1}^N [\vb{C}^2]_{v, v}\\
    &=& \Tr(\vb{C}^2).
\end{eqnarray}

And so, 

\begin{equation}
    \E\left[\frac{\ga \vb{a}_k, \vb{a}_j \dr^2}{||\vb{a}_k||^2 ||\vb{a}_j||^2}\right] = \frac{\tau(\vb{C}^2)}{N}.
\end{equation}

Furthermore, Cauchy–Schwarz inequality easily proves the remainder of the theorem.

\end{proof}

\begin{lemma}\label{lemme1}
Let $\vb{C}$ be a rotationally invariant covariance matrix such that $\tau(\vb{C})=1$. We write $\vb{C}$ as
\begin{equation}
    \begin{pmatrix}
        \vb{C}_{11} & \vb{b}^T\\
        \vb{b} & \vb{\hat{C}}
    \end{pmatrix}
\end{equation}
with $\vb{C}_{11}\in \R, \vb{b}\in \R^{N-1}$ and $\vb{\hat{C}}$ is a $(N-1)\times (N-1)$ matrix. Then, 
    $\vb{b}$ is a Gaussian vector with a squared-norm of $\tau(\vb{C}^2)-1$.
\end{lemma}

\begin{proof}
If we take 

\begin{equation}
    \vb{O} = \begin{pmatrix}
        1 & \vb{0}\\
        \vb{0} & \vb{\hat{O}}
    \end{pmatrix}
\end{equation}
where $\vb{\hat{O}}$ is a $(N-1)\times (N-1)$ orthogonal matrix, we have

\begin{equation}
    \vb{O}\vb{C} =
    \begin{pmatrix}
        \vb{C}_{11} & \vb{b}^T\\
        \vb{\hat{O}}\vb{b} & \vb{\hat{O}}\vb{\hat{C}}
    \end{pmatrix}.
\end{equation}

As $\vb{C}$ is a rotationally invariant matrix, we can deduce that $\vb{b}$ is likewise a rotationally invariant vector, establishing it as a Gaussian vector. What is its Euclidean norm? One can check that

\begin{equation}
    [\vb{C}^2]_{11} = \vb{C}_{11}^2 + ||\vb{b}||^2.
\end{equation}

In the large $N$ limit, $\vb{C}_{11} = \tau(\vb{C}) = 1$ and $[\vb{C}^2]_{11} = \tau(\vb{C}^2)$. Hence, 

\begin{equation}
    ||\vb{b}||^2 = \tau(\vb{C}^2) - 1.
\end{equation}

\end{proof}

\begin{lemma}\label{lemme2}
Let $\vb{C}$ be a rotationnaly invariant covariance matrix defined by its transform $S_{\vb{C}}$, $\Sigma$ the bottom-right $(N-1)\times (N-1)$ block of $\vb{C}^{-1}$ and $\vb{e}_1, ..., \vb{e}_T$ $T$ independent vectors drawn from $\mathcal{N}(0, \vb{\Sigma})$. Let's define the $(N-1)\times T$ matrix $\vb{E}$ as
\begin{equation}
    \vb{E} := (\vb{e}_k)_{1\leq k\leq T}.
\end{equation}

So, we can express the S-transform of $\vb{Y} = \frac{\vb{E}^T\vb{E}}{N}$ in terms of $S_{\vb{C}}$:
\begin{equation}
        S_{\vb{Y}}(z) = \frac{S_{\vb{C}}\left(\frac{z}{q}\right)}{1 + \frac{z}{q}}.
\end{equation}
\end{lemma}

\begin{proof}
We have

\begin{eqnarray}
    \mathfrak{t}_{\vb{Y}}(z)&=& \frac{1}{T}\Tr\left[\vb{Y}\left(z\vb{1} - \vb{Y}\right)^{-1}\right]\\
    &=& \frac{1}{T}\sum_{n=0}^{+\infty} \frac{1}{z^{n+1}} \Tr(\vb{Y}^{n+1})\\
    &=& \frac{1}{T}\sum_{n=0}^{+\infty} \frac{1}{z^{n+1}} \Tr\left[\left(\frac{\vb{E}^T \vb{E}}{N}\right)^{n+1}\right]\\
    &=& q \sum_{n=0}^{+\infty} \frac{1}{(qz)^{n+1}} \frac{1}{N} \Tr\left[\left(\frac{\vb{E} \vb{E}^T}{T}\right)^{n+1}\right]\\
    &=& q \mathfrak{t}_{\frac{\vb{E}\vb{E}^T}{N}}(qz)
\end{eqnarray}
and so, 

\begin{equation}
    \mathfrak{t}^{-1}_{\vb{Y}}(z) = \frac{1}{q} \mathfrak{t}^{-1}_{\frac{\vb{E} \vb{E}^T}{T}}\left(\frac{z}{q}\right).
\end{equation}

\ref{definitionS} yields

\begin{eqnarray}
    S_{\vb{Y}}(z) &=& \frac{z+1}{z \mathfrak{t}^{-1}_{\vb{Y}}(z)}\\
    &=& \frac{1+z}{z/q \mathfrak{t}^{-1}_{\frac{\vb{E} \vb{E}^T}{T}}\left(\frac{z}{q}\right)}\\
    &=& \frac{1+z}{1 + z/q} S_{\frac{\vb{E} \vb{E}^T}{T}}\left(\frac{z}{q}\right)\\
\end{eqnarray}

$\frac{\vb{E} \vb{E}^T}{T}$ is colored Wishart matrix and we know that

\begin{equation}
    S_{\frac{\vb{E} \vb{E}^T}{T}}(z)(q) = \frac{S_{\vb{\Sigma}(z)}}{1 + qz} = \frac{S_{\vb{C}(z)}}{1 + qz}.
\end{equation}

Indeed, $\vb{\Sigma}^{-1}$ is the bottom-right $(N-1)\times (N-1)$ block of $\vb{C}^{-1}$. In the large $N$ limit, $\vb{\Sigma}^{-1}$ and $\vb{C}^{-1}$ have the same behavior and therefore share the same $S$-transform and consequently, so do $\vb{\Sigma}$ and $\vb{C}$.

Finally, we have proved that

\begin{equation}
    S_{\vb{Y}}(z) = \frac{1+z}{1 + z/q} \frac{S_{\vb{C}}\left(\frac{z}{q}\right)}{1 + z} = \frac{S_{\vb{C}}\left(\frac{z}{q}\right)}{1 + \frac{z}{q}}.
\end{equation}

\end{proof}

\begin{lemma}
    Let $\vb{C}$ be a Wishart matrix of parameter $p$. We write $\vb{C}$ as
\begin{equation}
    \begin{pmatrix}
        \vb{C}_{11} & \vb{b}^T\\
        \vb{b} & \vb{\hat{C}}
    \end{pmatrix}
\end{equation}
with $\vb{C}_{11}\in \R, \vb{b}\in \R^{N-1}$ and $\vb{\hat{C}}$ is a $(N-1)\times (N-1)$ matrix. Let's define the $(N-1)\times (N-1)$ matrix as

\begin{equation}
    \vb{\Sigma} :=  \hat{\vb{C}} - \frac{\vb{b} \vb{b}^T}{\vb{C}_{1 1}}.
\end{equation}

So, $\vb{\Sigma}$ and $\vb{b}$ are independent random variables.

\end{lemma}

\begin{proof}   

It's a well-known fact that

\begin{equation}\label{distribution}
    P(\vb{C}) \propto (\det \vb{C})^{(T-N-1)/2)} \exp\left[\frac{-T}{2} \Tr(\vb{C})\right]
\end{equation}

By using the coefficient $C_{11}$ as the pivot, we can eliminate $\vb{b}$ from the first column while preserving the determinant, as transvections have a determinant of $1$. Thus, 

\begin{eqnarray}
    \det(\vb{C}) &=& \begin{vmatrix}
        \vb{C}_{11} & \vb{b}^T\\
        \vb{b} & \Tilde{C}.
    \end{vmatrix}\\
    &=&\begin{vmatrix}
        \vb{C}_{11} & \vb{b}^T\\
        \vb{0} & \Tilde{C} - \frac{\vb{b}\vb{b}^T}{\vb{C}_{11}}
    \end{vmatrix}\\
    &=& \vb{C}_{11} \det(\vb{\Sigma})\label{determinant}.
\end{eqnarray}

Furthermore, 

\begin{eqnarray}\label{trace}
    \Tr(\vb{C}) &=& \vb{C}_{11} + \Tr(\vb{\hat{C}})\\
    &=& \vb{C}_{11} + \Tr(\vb{\Sigma}) + \frac{1}{\vb{C}_{11}} ||\vb{b}||^2.
\end{eqnarray}

Using the structure of \eqref{distribution} with \eqref{determinant}, \eqref{trace} and the fact that the transformation
\begin{equation}
    (\vb{C}_{11}, \vb{b}, \vb{\Tilde{C}}) \mapsto \left(\vb{C}_{11}, \vb{b}, \vb{\Sigma} = \vb{\Tilde{C}} - \frac{\vb{b}\vb{b}^T}{\vb{C}_{11}}\right)
\end{equation}has a Jacobian equal to $1$, we obtain the result.

\end{proof}

\begin{lemma}\label{integrals}
    Let's define $I_1$ and $I_2$ as
    \begin{eqnarray}
    I_1(Z) &=& \E\left[\frac{a^2 f^2(a^2)}{Z - f(a^2)} + a^2 f(a^2)\right]\\
    I_2(Z) &=&\E\left[\frac{f(a^2)}{Z - f(a^2)}\right]
\end{eqnarray}
with $a$ a standard Gaussian variable and $f$ the function $y\mapsto 1 - \frac{1}{y}$.
So, 

\begin{eqnarray}
    I_1(1) &=&2\\
    I_2(1) &=&0\\
    I_1'(1) & = & -10\\
    I_2'(1) &=& -2.
\end{eqnarray}
\end{lemma}

\begin{proof}   
We have
\begin{eqnarray}
    I_1(1) &=& \E\left[\frac{a^2 f^2(a^2)}{1-f(a^2)}\right]\\
    &=& \E(a^4 f^2(a^2))\\
    &=& \E((a^2 - 1)^2)\\
    &=& \E(a^4) - 2\E(a^2) + 1\\
    &=& 3 - 2 + 1\\
    &=& 2.
\end{eqnarray}

\begin{eqnarray}
    I_2(1) &=& \E\left[\frac{f(a^2)}{1-f(a^2)}\right]\\
    &=& \E(a^2 f(a^2))\\
    &=& \E(a^2-1)\\
    &=& 1 - 1\\
    &=&0.
\end{eqnarray}

We know that

\begin{eqnarray}
    I_1(Z) &=&\E\left[\frac{a^2 \left(f(a^2)\right)^2}{Z-f(a^2)}\right]\\
      &=& \E\left[\frac{a^2 \left(1 - \frac{1}{a^2}\right)^2}{Z-1 + \frac{1}{a^2}}\right]\\
     &=& \E\left[\frac{(a^2 - 1)^2}{Za^1 - a^2 + 1}\right]
\end{eqnarray}
and so, 

\begin{equation}
    I_1'(Z) = \E\left[\frac{-a^2(a^2 - 1)^2}{(Za^1 - a^2 + 1)^2}\right].
\end{equation}

Hence, 

\begin{eqnarray}
    I_1'(1) &=& \E[-a^2(a^2 - 1)^2]\\
    &=& -\E(a^6) + 2 \E(a^4) - \E(a^2)\\
    &=& - 15 + 2.3 - 1\\
    &=& -10.
\end{eqnarray}

In the same way, we have

\begin{eqnarray}
    I_2(Z) &=&\E\left[\frac{\left(f(a^2)\right)}{Z-f(a^2)}\right]\\
      &=& \E\left[\frac{a^2 \left(1 - \frac{1}{a^2}\right)}{Z-1 + \frac{1}{a^2}}\right]\\
     &=& \E\left[\frac{(a^2 - 1)}{Za^1 - a^2 + 1}\right]
\end{eqnarray}
and so,

\begin{align}
    I_2'(1) &= \E(-a^2(a^2 - 1))\\
    &= -\E(a^4) + \E(a^2) \\
    &= - 3 + 1 = -2
\end{align}

\end{proof}

\section*{Acknowledgment}

The authors would like to thank Federico Ricci-Tersenghi for pointing out the error in \cite{potters2020first} (corrected in \cite{erratum}) which made them think again about the phase retrieval problem.

\ifCLASSOPTIONcaptionsoff
  \newpage
\fi

\bibliographystyle{IEEEtran}
\bibliography{References}

\end{document}